\documentclass[letterpaper, 11pt, onecolumn]{article}
\pdfoutput=1 

\usepackage{lmodern,booktabs,authblk,tocloft}
\usepackage[margin=1in]{geometry}
\usepackage{graphicx}
\usepackage[linesnumbered,ruled,vlined]{algorithm2e}
\SetKwInput{KwInput}{Input}                
\SetKwInput{KwOutput}{Output}
\SetKwInOut{KwPromise}{Promise}
\SetKwFor{RepTimes}{repeat}{times}{end}
\usepackage{complexity}
\usepackage{amsmath}
\usepackage{fullpage}
\usepackage{amsfonts}
\usepackage{amssymb}
\usepackage{dsfont}
\usepackage{amsthm}
\usepackage{thm-restate}
\usepackage{tensor}
\usepackage[dvipsnames]{xcolor}
\usepackage{comment}
\usepackage{mathtools}
\usepackage{physics}
\usepackage{bbold}
\usepackage{appendix}
\usepackage{multirow}
\usepackage[ruled,vlined]{algorithm2e}
\usepackage[colorlinks]{hyperref}
\hypersetup{
pdfstartview={FitH},
pdfnewwindow=true,
colorlinks=true,
linkcolor=MidnightBlue,
citecolor=MidnightBlue,
filecolor=MidnightBlue,
urlcolor=MidnightBlue}
\usepackage{enumitem}
\usepackage[capitalize]{cleveref}

\makeatletter
\DeclareRobustCommand\bfseries{%
  \not@math@alphabet\bfseries\mathbf
  \fontseries\bfdefault\selectfont\boldmath}
\makeatother

\def\Tr{\mathrm{Tr}}

\def\dqc1{\textsc{DQC1}}
   
\newcommand{\F}{\mathbb F}
\renewcommand{\C}{\mathbb C}
\renewcommand{\tr}{\operatorname{tr}}
\newcommand{\End}{\operatorname{End}}
\newcommand{\PSL}{\operatorname{PSL}}
\renewcommand{\SL}{\operatorname{SL}}
\newcommand{\HS}{\mathrm{HS}}
\newcommand{\PGL}{\operatorname{PGL}}

\newtheorem{theorem}{Theorem}[section]
\newtheorem{lemma}[theorem]{Lemma}

\newtheorem{remark}[theorem]{Remark}
\newtheorem{corollary}[theorem]{Corollary}

\title{Ramanujan quantum expanders from the Weil representation}

\author{Siddhartha Jain\\
\vspace{-0.5em}\small\itshape The University of Texas at Austin}
\date{\today}

\begin{document}

\maketitle

\begin{abstract}
For every odd prime power $q$ and $D=q+1$, we construct an infinite family of Ramanujan quantum expanders of degree $D$. The construction transfers Morgenstern's Ramanujan Cayley graphs on $\PSL_2(\F_{p})$ for $p$ which is an even power of $q$, through the odd irreducible subrepresentation of the Weil representation of $\SL_2(\F_{p})$. For a quantum expander of dimension $N$, our implementation uses $O(\log^2 N)$ elementary gates, and $O(\log N)$ ancilla qudits, using a fixed finite gate set depending on $q$. An advantage compared to the previous work of Iyer, Jain, Jordan, and Somma ({\footnotesize FOCS 2026}) is that assuming the quantum circuit is implemented exactly, we satisfy the $2\sqrt{D-1}/D$ singular value bound exactly without any additive error. 
\end{abstract}


Quantum expanders are analogs of \emph{random walks} over classical expander graphs. They are unital channels for which the second singular value is small.
We study mixed-unitary channels, which choose one of $D$ unitaries uniformly and apply it to the input state, where $D$ is called the degree.
For even $D\geq4$, Hastings showed that choosing $D/2$ independent Haar-random unitaries and pairing each unitary with its adjoint gives channels whose second singular value converges
to $2\sqrt{D-1}/D$ as the dimension grows~\cite{hastings07}. The challenge is to construct these channels explicitly, with efficient circuits.

Harrow observed that a classical Cayley expander on a group $G$ with symmetric generating set $S$, together with an $N$-dimensional irreducible representation $\mu$ of $G$, gives the quantum channel
\begin{equation}\label{eq:harrow-transfer-channel}
\Phi_\mu(X):=\frac1{|S|}\sum_{s\in S}\mu(s)X\mu(s)^*.
\end{equation}
The second singular value of the channel in \cref{eq:harrow-transfer-channel} is no larger than the nontrivial spectral norm of the classical walk~\cite{Harrow2008QuantumExpanders}. So we immediately get many quantum channels that satisfy the analytic bound, but the crux is implementing $\mu(s)$ in time polynomial in $\log N$. Harrow posed the remaining gap to be to design efficient quantum Fourier transforms (QFT) over groups with Ramanujan Cayley expanders.

But he also noted that one does not require a QFT. A nonabelian quantum Fourier transform over a group $G$ with Ramanujan Cayley expanders is definitely sufficient: it decomposes the regular representation into all of its irreducible blocks, after which one can read off $\mu(s)$. However, for $G=\PSL_2(\F_p)$ no efficient group QFT is known, and this is the obstruction Harrow identified for the LPS-type example~\cite{lubotzky1988ramanujan,Harrow2008QuantumExpanders}.
But the group QFT solves a much larger problem than we need. It constructs every irreducible representation at once, whereas the quantum expander construction requires only one large irreducible representation.

For $\PSL_2(\F_p)$, we use the odd Weil representation.  It is obtained from a unitary action of $\SL_2(\F_p)$ on the $p$-dimensional Hilbert space $\mathcal K_p=\C^{\F_p}=\operatorname{span}\{\ket{x}:x\in\F_p\}$ whose computational basis is indexed by the elements of the finite field.  The elementary matrices generating $\SL_2(\F_p)$ have particularly simple actions on this basis: they apply a phase depending quadratically on $x$, relabel $\ket{x}$ by multiplying $x$ by a nonzero field element, or apply the Fourier transform over the additive group of $\F_p$~\cite{Gerardin1977,GurevichHadaniSochen2008}.
It remains to implement an arbitrary matrix in $\SL_2(\F_p)$.  No general group machinery is needed.  One can check that
\begin{equation}\label{eq:matrix-factorization}
\begin{aligned}
g=\begin{pmatrix}a&b\\c&d\end{pmatrix}
&=
\begin{pmatrix}1&a/c\\0&1\end{pmatrix}
\begin{pmatrix}0&1\\-1&0\end{pmatrix}
\begin{pmatrix}-c&0\\0&-c^{-1}\end{pmatrix}
\begin{pmatrix}1&d/c\\0&1\end{pmatrix},
&& c\neq 0,\\[1ex]
g=\begin{pmatrix}a&b\\0&a^{-1}\end{pmatrix}
&=
\begin{pmatrix}a&0\\0&a^{-1}\end{pmatrix}
\begin{pmatrix}1&b/a\\0&1\end{pmatrix},
&& c=0.
\end{aligned}
\end{equation}
Here we used $\frac{ad-1}{c} = b$. Thus, every $g\in\SL_2(\F_p)$ is a product of at most four matrices of the three forms appearing above.  Under the Weil representation these three forms become a quadratic phase, the finite-field Fourier transform, and a field scaling, respectively.  Hence, implementing an arbitrary generator reduces to implementing these three elementary quantum operations.

Two representation-theoretic facts make this implementation fit Harrow's transfer theorem. First, the odd parity subspace has dimension $N = (p-1)/2$ and is irreducible.
Second, each element of $\PSL_2(\F_p)$ has two lifts to $\SL_2(\F_p)$, and the corresponding Weil operators differ only by
a scalar phase on the odd subspace.
These phases disappear under conjugation and therefore do not affect
the channel.
With this in mind, we can state the argument succinctly. We can combine the Morgenstern classical Ramanujan Cayley graphs, the odd Weil representation which is one large directly implementable irreducible representation, and Harrow's transfer principle which implies the quantum spectral bound. With careful examination one can construct a $O(\log^2 N)$ size circuit for a Ramanujan quantum expander, with $\widetilde{O}(\log^5 N)$ classical preprocessing time. We defer the discussion of the details to the appendices, which assumes minimal background in quantum computing and representation theory.

\paragraph{Related work.} The work of Jeronimo, Mittal, Roy, and Wigderson constructed quantum expanders~\cite{OpAmp25}. Starting from a
constant-degree family with a constant gap, their construction achieves
contraction at most $\lambda$ with degree
$D=O(\lambda^{-(2+o_\lambda(1))})$, but it does not give the optimal fixed-degree bound we achieve here. The recent work using the quantum Hermite transform gets $\epsilon$-close to the Ramanujan bound in $\polylog (1/\varepsilon)$ time for degree 6~\cite{jain2025efficientquantumhermite,iyer2026efficientquantumcircuitshighdimensional}, but our circuit satisfies the bound without additive error.
While it is hard to imagine implementing such a quantum circuit exactly in practice, due to the Solovay–Kitaev theorem the overhead for implementing the circuit with $\varepsilon/2$ operator norm error is $O(\log^c (\log^2 N/\varepsilon))$ where $c\leq 1.441$~\cite{Kuperberg2023BreakingCubic}. This gives the additive error $\varepsilon$ in the singular value if we additionally approximate the unitaries on the odd subspace (\cref{sec:approximate-implementation}).
We do not know how to get this scaling with the approach of \cite{iyer2026efficientquantumcircuitshighdimensional}.

\section*{Open problems}

\begin{enumerate}
    \item Can we construct Ramanujan quantum expanders for all degrees?
    
    \item Can we construct fixed-degree Ramanujan quantum expanders for every sufficiently large Hilbert-space dimension, rather than only for the dimensions $N_k=(q^{2k}-1)/2$? The construction of~\cite{iyer2026efficientquantumcircuitshighdimensional} gives near-Ramanujan quantum expanders in every dimension.
\end{enumerate}

\subsection*{Acknowledgements}

The author thanks Aram Harrow for discussions about the exposition of this work. The author was supported by the Amazon AI fellowship. Generative AI was used to search the literature for an quantum-efficient representation of $\SL_2(\F_p)$, and it was also used by the author understand the Weil representation and aid with writing the appendices.

\appendix

\renewcommand*{\theHtheorem}{\thesection.\arabic{theorem}}
\renewcommand*{\theHlemma}{\theHtheorem}
\renewcommand*{\theHremark}{\theHtheorem}
\renewcommand*{\theHcorollary}{\theHtheorem}

\section{Quantum expansion and the transfer lemma}\label{sec:preliminaries}

We first define the contraction that a quantum expander must satisfy. We then prove the transfer lemma: a spectral bound for a classical Cayley walk gives the same bound for an irreducible projective representation. The proof compares norms directly, without decomposing a representation into irreducible blocks.

\paragraph{Parameters.}
Throughout the appendices, we fix an odd prime power $q=\ell^r$ and the degree $D=q+1$. We let $n$ range over the positive even integers and put
\[
    p=q^n=\ell^m,\qquad m=rn,\qquad N=\frac{p-1}{2}.
\]
Thus $q$ fixes the degree, while $n$ controls the dimension. The symbol $p$ denotes a field size, not necessarily a prime. We prove the spectral bound in \cref{sec:quantum-expander}. We give the circuits and their classical preprocessing in \cref{sec:quantum-circuit}.

\subsection{The contraction we need}

For a finite set $X$, we write $\C^X$ for the Hilbert space with orthonormal basis $\{|x\rangle:x\in X\}$. Thus $\C^{\F_p}$ is a register whose basis labels are field elements. For a finite-dimensional Hilbert space $\mathcal H$, let $\End(\mathcal H)$ be its space of linear operators. We use the \emph{Hilbert--Schmidt inner product} and norm
\[
    \langle X,Y\rangle_{\HS}:=\tr(X^*Y),\qquad
    \|X\|_{\HS}:=\sqrt{\tr(X^*X)}.
\]
The traceless operators form the subspace
\[
    \End_0(\mathcal H):=\{X\in\End(\mathcal H):\tr X=0\}.
\]
This is the orthogonal complement of the identity, since $\langle\mathds{1},X\rangle_{\HS}=\tr X$.

A \emph{density matrix} is a positive semidefinite operator of trace one. If $\dim\mathcal H=N$, the \emph{maximally mixed state} is $\omega_N:=\mathds{1}/N$. A degree-$D$ \emph{mixed-unitary channel} chooses one of $D$ unitaries uniformly and applies it:
\begin{equation*}
    \Phi(X):=\frac1D\sum_{s=1}^D U_sXU_s^*.
\end{equation*}
This map takes density matrices to density matrices and fixes $\omega_N$. We call it an $(N,D,\lambda)$ \emph{quantum expander} if
\begin{equation}\label{eq:density-contraction}
    \|\Phi(\rho)-\omega_N\|_{\HS}
    \leq\lambda\|\rho-\omega_N\|_{\HS}
\end{equation}
for every density matrix $\rho$. It is \emph{Ramanujan} when $\lambda\leq\lambda_D$, where
\begin{equation*}
    \lambda_D:=\frac{2\sqrt{D-1}}D.
\end{equation*}
We will use the equivalent bound on traceless operators.

\begin{lemma}[Spectral form of quantum expansion]\label{lem:spectral-form}
For a mixed-unitary channel $\Phi$, the least $\lambda$ for which \cref{eq:density-contraction} holds is
\begin{equation*}
    \sigma_2(\Phi):=
    \left\|\Phi\big|_{\End_0(\mathcal H)}\right\|_{\HS\to\HS}.
\end{equation*}
If $\sigma_2(\Phi)<1$, then $\omega_N$ is the unique fixed density matrix.
\end{lemma}
\begin{proof}
We first pass from density matrices to Hermitian traceless operators, then to all traceless operators. Write $\rho=\omega_N+X$. Then $X$ is Hermitian and traceless. Conversely, for every Hermitian traceless $X$, the operator $\omega_N+tX$ is a density matrix for all sufficiently small $t>0$. Thus \cref{eq:density-contraction} is equivalent to $\|\Phi(X)\|_{\HS}\leq\lambda\|X\|_{\HS}$ for Hermitian traceless $X$.

For a general traceless $X$, write $X=A+iB$, with $A,B$ Hermitian and traceless. Since $\Phi$ preserves Hermitian operators,
\[
    \|X\|_{\HS}^2=\|A\|_{\HS}^2+\|B\|_{\HS}^2,\qquad
    \|\Phi(X)\|_{\HS}^2
    =\|\Phi(A)\|_{\HS}^2+\|\Phi(B)\|_{\HS}^2.
\]
Applying the Hermitian bound to $A$ and $B$ proves the claimed norm formula. If $\rho$ is fixed and $\sigma_2(\Phi)<1$, applying the strict contraction to $\rho-\omega_N$ gives $\rho=\omega_N$.
\end{proof}

\subsection{Cayley walks and irreducibility}

Let $S=S^{-1}$ be a generating set of a finite group $G$. The normalized right \emph{Cayley walk} acts on functions $f:G\to\C$ by
\begin{equation*}
    (A_Sf)(g):=\frac1{|S|}\sum_{s\in S}f(gs).
\end{equation*}
We give these functions the norm $\|f\|_2^2:=\sum_{g\in G}|f(g)|^2$. Symmetry of $S$ makes $A_S$ Hermitian. Generation makes its eigenvalue-one space exactly the constant functions. Writing $\mathbf1$ for the constant-one function, the classical estimate we need is
\begin{equation}\label{eq:generic-cayley-bound}
    \left\|A_S\big|_{\mathbf1^\perp}\right\|_{2\to2}\leq\lambda.
\end{equation}
Here $\mathbf1^\perp$ consists of the mean-zero functions. A connected non-bipartite Ramanujan Cayley graph of degree $D$ satisfies this estimate with $\lambda=\lambda_D$.

A \emph{unitary representation} assigns a unitary $U_g$ to each $g\in G$, with $U_gU_h=U_{gh}$. A \emph{projective unitary representation} allows a scalar phase:
\[
    U_gU_h=c(g,h)U_{gh},\qquad |c(g,h)|=1.
\]
It is \emph{irreducible} if no subspace other than $0$ and $\mathcal H$ is invariant under every $U_g$. The \emph{commutant} is the set of operators that commute with every $U_g$. We need the following form of Schur's lemma.

\begin{lemma}[The commutant criterion]\label{lem:schur-commutant}
A finite-dimensional unitary representation, ordinary or projective, is irreducible exactly when its commutant consists of the scalar operators.
\end{lemma}
\begin{proof}
Suppose the representation is irreducible and $T$ commutes with every $U_g$. Choose an eigenvalue $a$ of $T$. Its eigenspace is nonzero and invariant, so it is all of $\mathcal H$. Hence $T=a\mathds{1}$. Conversely, if a nonzero proper subspace is invariant, unitarity makes its orthogonal complement invariant too. The orthogonal projector onto that subspace is a nonscalar operator in the commutant. Neither argument depends on the projective phases.
\end{proof}

\subsection{The transfer lemma}

The idea is to turn a traceless operator $X$ into the function $g\mapsto U_gXU_g^*$. Irreducibility makes this function mean zero. The classical walk then contracts it. This proves Harrow's transfer principle~\cite{Harrow2008QuantumExpanders} in the projective form we use.

\begin{lemma}[Spectral transfer]\label{lem:spectral-transfer}
Let $S=S^{-1}$ generate a finite group $G$, and suppose its normalized right Cayley walk satisfies \cref{eq:generic-cayley-bound}. Let $g\mapsto U_g$ be an irreducible projective unitary representation on $\mathcal H$. Then
\[
    \Phi(X):=\frac1{|S|}\sum_{s\in S}U_sXU_s^*
\]
satisfies $\|\Phi(X)\|_{\HS}\leq\lambda\|X\|_{\HS}$ for every traceless operator $X$.
\end{lemma}
\begin{proof}
For an operator $X$, define the operator-valued function
\[
    f_X(g):=\frac1{\sqrt{|G|}}U_gXU_g^*.
\]
We give such functions the norm $\|f\|^2:=\sum_{g\in G}\|f(g)\|_{\HS}^2$. Unitary conjugation preserves the Hilbert--Schmidt norm, so
\begin{equation}\label{eq:transfer-isometry}
    \|f_X\|^2
    =\frac1{|G|}\sum_{g\in G}\|U_gXU_g^*\|_{\HS}^2
    =\|X\|_{\HS}^2.
\end{equation}

We next show that $f_X$ has mean zero when $\tr X=0$. The average
\[
    \overline X:=\frac1{|G|}\sum_{g\in G}U_gXU_g^*
\]
commutes with every $U_h$. Indeed, conjugating this sum by $U_h$ permutes its terms, since the projective phases cancel. By \cref{lem:schur-commutant}, $\overline X$ is scalar. It has trace zero, so it is zero.

Finally, applying the classical walk to $f_X$ gives the same function as applying the quantum channel to $X$:
\[
\begin{aligned}
    (A_Sf_X)(g)
    &=\frac1{|S|\sqrt{|G|}}\sum_{s\in S}U_{gs}XU_{gs}^*\\
    &=\frac1{\sqrt{|G|}}U_g\Phi(X)U_g^*
    =f_{\Phi(X)}(g).
\end{aligned}
\]
The bound in \cref{eq:generic-cayley-bound} also holds for operator-valued mean-zero functions: apply it to each coordinate in an orthonormal basis of operator space and sum the squared bounds. Thus \cref{eq:transfer-isometry} gives
\[
    \|\Phi(X)\|_{\HS}
    =\|f_{\Phi(X)}\|
    =\|A_Sf_X\|
    \leq\lambda\|f_X\|
    =\lambda\|X\|_{\HS}.
\]
\end{proof}

The function $f_X$ is used only in the analysis. We do not prepare it, enumerate $G$, or compute a group Fourier transform.

\section{The classical graph family}\label{sec:classical-input}

We need $D=q+1$ generators whose Cayley walk has norm at most $\lambda_D$ on mean-zero functions. Morgenstern's construction supplies them~\cite{MORGENSTERN199444}. For a publicly accessible reference we use the 1991 report~\cite{Morgenstern1991Report}. We use theorem numbering for this version below. Here $\SL_2(K)$ is the group of determinant-one matrices over a field $K$, and $\PSL_2(K):=\SL_2(K)/\{\pm\mathds{1}_2\}$.

\begin{theorem}[Morgenstern finite-field $\PSL_2$ family]\label{thm:morgenstern-family}
Fix $D=q+1$, where $q=\ell^r$ is an odd prime power. For every even $n\geq2$, put $p:=q^n$. There is a symmetric generating set $S_p\subseteq G_p:=\PSL_2(\F_p)$ such that $|S_p|=D$, $1\notin S_p$, and $\operatorname{Cay}(G_p,S_p)$ is connected and non-bipartite. Its normalized right walk
\begin{equation*}
    (A_pf)(g):=\frac1D\sum_{s\in S_p}f(gs)
\end{equation*}
satisfies
\begin{equation}\label{eq:classical-bound}
    \left\|A_p\big|_{\mathbf1^\perp}\right\|_{2\to2}\leq\lambda_D.
\end{equation}
A field model for $\F_p$, the generators $S_p$, and lifts to $\SL_2(\F_p)$ are computable deterministically in $O(n^5)$ operations over the fixed field $\F_q$.
\end{theorem}

We prove this theorem below. The main choice is between two finite quotients: one gives a bipartite graph, while the other gives the norm bound we need. We defer the algorithm that constructs the field model to \cref{sec:classical-preprocessing}.

\subsection{Which finite quotient?}

For an odd prime power $Q$, the multiplicative group $\F_Q^*$ is cyclic. A nonzero element is a \emph{square} if it has a square root in the field. Euler's test gives
\begin{equation*}
    a^{(Q-1)/2}=
    \begin{cases}
        1,&a\in(\F_Q^*)^2,\\
        -1,&a\notin(\F_Q^*)^2.
    \end{cases}
\end{equation*}
The \emph{square class} of $a$ records which of these two cases holds.

Write $\PGL_2(K):=\operatorname{GL}_2(K)/(K^*\mathds{1}_2)$ for invertible matrices modulo nonzero scalars. Multiplying a matrix by a scalar changes its determinant by a square. Thus determinant induces the map
\begin{equation*}
\begin{aligned}
    \det_{\square}:\PGL_2(K)&\longrightarrow K^*/(K^*)^2,
    & [A]&\longmapsto[\det A],\\
    \ker(\det_{\square})&=\PSL_2(K).
\end{aligned}
\end{equation*}
An element of $\PSL_2(K)$ has two lifts to $\SL_2(K)$, differing by $-\mathds{1}_2$.

Let $g\in\F_q[t]$ be monic and irreducible. Its \emph{residue field} is $k_g:=\F_q[t]/(g)$, and $\overline t$ denotes the image of $t$ in this field. A \emph{projective quaternion class} is an invertible quaternion modulo nonzero scalars. Quaternion conjugation is denoted $x\mapsto\overline x$. The \emph{reduced norm} $\operatorname{Nrd}(x):=x\overline x$ plays the role of determinant and gives $x^{-1}=\overline x/\operatorname{Nrd}(x)$ when the norm is nonzero.

We specify the basic generators so that the signs and the compiler's inputs are explicit. Fix a nonsquare $\nu\in\F_q^*$ and use the quaternion algebra over $\F_q(t)$ with
\[
    i^2=\nu,\qquad j^2=t-1,\qquad ij=-ji.
\]
Set
\begin{equation}\label{eq:basic-norm-t-generators}
    \mathcal T_q:=\{(c,d)\in\F_q^2:\nu d^2-c^2=1\},
    \qquad \xi_{c,d}:=1+cj+dij.
\end{equation}
These are Morgenstern's basic norm-$t$ elements~\cite[Lemma~4.2(c), Eq.~(9), and Definition~4.3]{Morgenstern1991Report}. There are $q+1$ pairs in $\mathcal T_q$: the field norm $\F_{q^2}^*\to\F_q^*$ is onto, each fiber has $q+1$ elements, and $\mathcal T_q$ is its fiber over $-1$, written as $c^2-\nu d^2=-1$. Directly,
\[
    \operatorname{Nrd}(\xi_{c,d})
    =1+(\nu d^2-c^2)(t-1)=t,
    \qquad \overline{\xi_{c,d}}=\xi_{-c,-d}.
\]
Thus conjugation pairs these projective classes with their inverses. The sign in $\nu d^2-c^2=1$ matters: replacing its left side by $c^2-\nu d^2$ would give norm $2-t$, not $t$.

\begin{theorem}[Morgenstern's finite quotient, specialized]\label{thm:morgenstern-specialization}
Let $g\in\F_q[t]$ be monic and irreducible, relatively prime to $t(t-1)$, and of even degree. Reducing the projective classes of the $q+1$ elements in \cref{eq:basic-norm-t-generators} modulo $g$ gives a symmetric set $B_g$ of distinct nonidentity generators of a finite quotient $\Gamma_g$, with
\begin{equation}\label{eq:morgenstern-quotient-dichotomy}
    \Gamma_g\cong
    \begin{cases}
        \PSL_2(k_g),&\overline t\in(k_g^*)^2,\\
        \PGL_2(k_g),&\overline t\notin(k_g^*)^2.
    \end{cases}
\end{equation}
The Cayley graph $\operatorname{Cay}(\Gamma_g,B_g)$ is $(q+1)$-regular and Ramanujan.
\end{theorem}
\begin{proof}
Theorem~4.10 of Morgenstern's report gives the finite $(q+1)$-regular Cayley graph, and Theorem~4.11 bounds its adjacency eigenvalues other than $\pm(q+1)$ by $2\sqrt q$ in absolute value~\cite[Theorems~4.10--4.11]{Morgenstern1991Report}.

To identify its group, note that
\[
    \nu^{(|k_g|-1)/2}
    =\left(\nu^{(q-1)/2}\right)^{1+q+\cdots+q^{\deg g-1}}
    =(-1)^{\deg g}=1.
\]
Thus $\nu$ has a square root in $k_g$, and the quaternion algebra has the matrix form in \cref{eq:quaternion-splitting}. Lemma~4.12 of the report shows that the reduction image contains $\PSL_2(k_g)$~\cite[Lemma~4.12]{Morgenstern1991Report}. Every basic generator has determinant $\overline t$, so the image in $\PGL_2(k_g)/\PSL_2(k_g)$ is generated by $[\overline t]$. This quotient has order two, proving \cref{eq:morgenstern-quotient-dichotomy}. Theorem~4.13 of the report states these two cases explicitly~\cite[Theorem~4.13]{Morgenstern1991Report}.
\end{proof}

\begin{remark}[Why we need the square case]\label{rem:morgenstern-branches}
If $\overline t$ is a nonsquare, every generator changes the determinant square class. These two classes give a bipartition of the Cayley graph on $\PGL_2(k_g)$. The function that is $1$ on one class and $-1$ on the other has mean zero and eigenvalue $-1$. A bipartite graph can be Ramanujan, but this eigenvalue prevents the absolute norm bound required by the transfer lemma. We choose $\overline t$ to be a square, which selects the $\PSL_2$ quotient.
\end{remark}

\subsection{Selecting the non-bipartite family}

We need a field model in which both $\overline t$ and the fixed nonsquare $\nu\in\F_q^*$ have known square roots. The first selects the desired quotient. The second lets us write its generators as matrices.

\begin{lemma}[Uniform square moduli]\label{lem:square-irreducible-moduli}
For every even $n\geq2$, a deterministic algorithm outputs polynomials $g_n,b_n,z_n\in\F_q[t]$ such that $g_n$ is monic irreducible of degree $n$, $\gcd(g_n,t(t-1))=1$, and, in $k_n:=\F_q[t]/(g_n)$,
\begin{equation*}
    b_n^2=\overline t,\qquad z_n^2=\nu.
\end{equation*}
It uses $O(n^5)$ arithmetic operations over the fixed field $\F_q$.
\end{lemma}

We prove this lemma in \cref{sec:classical-preprocessing}. We can now finish the graph argument.

\begin{proof}[Proof of \cref{thm:morgenstern-family}]
Apply \cref{lem:square-irreducible-moduli} and identify $\F_p$ with $k_n$. Since $\overline t=b_n^2$, \cref{thm:morgenstern-specialization} gives the quotient $G_p=\PSL_2(\F_p)$. Let $S_p$ be the image of $B_{g_n}$. It consists of $D$ distinct nonidentity symmetric generators and gives a connected Ramanujan Cayley graph.

It remains to rule out bipartiteness. A connected Cayley graph $\operatorname{Cay}(G,S)$ is bipartite only if there is a homomorphism $G\to\{\pm1\}$ that sends every generator to $-1$. Row elimination shows that $\PSL_2(\F_p)$ is generated by the images of the upper and lower triangular matrices with ones on the diagonal. Every nonidentity such matrix has the odd order $\ell$. Hence every homomorphism from $\PSL_2(\F_p)$ to a group of order two is trivial. The graph is therefore non-bipartite. Dividing Morgenstern's nontrivial adjacency bound $2\sqrt q$ by $D=q+1$ gives
\[
    \frac{2\sqrt q}{q+1}=\frac{2\sqrt{D-1}}D=\lambda_D,
\]
which proves \cref{eq:classical-bound}.

For the algorithmic claim, use $z_n^2=\nu$ to write Morgenstern's quaternion generators as matrices over $k_n$. Each matrix has determinant $\overline t=b_n^2$, so scaling it by $b_n^{-1}$ gives a lift to $\SL_2(\F_p)$. We give the formulas in \cref{eq:quaternion-splitting,eq:uniform-lifts}. Since $q$ is fixed, these steps use $O(n^3)$ operations after the field construction. The total is $O(n^5)$.
\end{proof}

\section{The odd Weil representation}\label{sec:weil}

We need a large irreducible representation whose operators are easy to apply. We begin with three operations on basis states labelled by $\F_p$: a quadratic phase, multiplication of the label by a nonzero field element, and a Fourier transform. These operations give the Weil representation on $\mathcal K_p:=\C^{\F_p}$. This representation is reducible. We restrict it to the states that change sign under $x\mapsto-x$, then prove that the restriction is irreducible.

\subsection{Three explicit operations}

To specify the phases, we use the field trace and its additive character:
\begin{equation*}
\begin{aligned}
    \Tr(z)&:=\operatorname{Tr}_{\F_p/\F_\ell}(z)
    =z+z^\ell+\cdots+z^{\ell^{m-1}},\\
    \psi(z)&:=\exp\!\left(\frac{2\pi i}{\ell}\Tr(z)\right).
\end{aligned}
\end{equation*}
In the exponential, we identify elements of $\F_\ell$ with their standard integer representatives. The next fact ensures that the phases distinguish field elements and that the Fourier transform is unitary.

\begin{lemma}[Trace pairing and character orthogonality]\label{lem:trace-pairing}
The map $\Tr:\F_p\to\F_\ell$ is a nonzero, hence surjective, $\F_\ell$-linear map. The pairing $(x,y)\mapsto\Tr(xy)$ is nondegenerate: for every $x\neq0$ there is a $y$ with $\Tr(xy)\neq0$. The function $\psi$ is a nontrivial additive character, and
\begin{equation}\label{eq:character-orthogonality}
    \frac1p\sum_{y\in\F_p}\psi(ay)=
    \begin{cases}1,&a=0,\\0,&a\neq0.\end{cases}
\end{equation}
\end{lemma}
\begin{proof}
The polynomial $T(X):=X+X^\ell+\cdots+X^{\ell^{m-1}}$ is nonzero and has degree less than $p$, so it cannot vanish at all $p$ field elements. Thus $\Tr$ is nonzero and therefore surjective. For $x\neq0$, choose $z$ with $\Tr(z)\neq0$ and set $y=x^{-1}z$. Then $\Tr(xy)\neq0$, which proves nondegeneracy.

Linearity of $\Tr$ gives $\psi(a+b)=\psi(a)\psi(b)$, so $\psi$ is an additive character. Surjectivity makes it nontrivial. For $a\neq0$, translating the sum in \cref{eq:character-orthogonality} by a point where $\psi(ay)\neq1$ multiplies it by a scalar other than one. The sum is therefore zero. The case $a=0$ is immediate.
\end{proof}

Let $\chi:\F_p^*\to\{\pm1\}$ be the quadratic character: it is $1$ on squares and $-1$ on nonsquares. Since $p$ is odd, $2^{-1}$ exists in $\F_p$.

\begin{theorem}[Finite-field Weil representation]\label{thm:weil-representation}
For $t\in\F_p$ and $u\in\F_p^*$, the matrices and corresponding operators
\begin{equation*}
\begin{aligned}
    n(t)&:=\begin{pmatrix}1&t\\0&1\end{pmatrix},
    & W_{n(t)}|x\rangle&:=\psi(tx^2/2)|x\rangle,\\[1ex]
    h(u)&:=\begin{pmatrix}u&0\\0&u^{-1}\end{pmatrix},
    & W_{h(u)}|x\rangle&:=\chi(u)|u^{-1}x\rangle,\\[1ex]
    w&:=\begin{pmatrix}0&1\\-1&0\end{pmatrix},
    & W_w|x\rangle&:=\frac1{\sqrt p}\sum_{y\in\F_p}\psi(xy)|y\rangle
\end{aligned}
\end{equation*}
extend, up to scalar phases, to a projective unitary representation $g\mapsto W_g$ of $\SL_2(\F_p)$ on $\mathcal K_p$.
\end{theorem}
\begin{proof}
The first operator is a diagonal phase, and the second is a phase times a basis permutation. The third is unitary by character orthogonality. To prove the projective multiplication law, we check how these operators conjugate translations and phase modulations.

For $a,b\in\F_p$, define
\[
    E_{a,b}|x\rangle:=\psi\bigl(a(x+b/2)\bigr)|x+b\rangle.
\]
An operator commuting with every $E_{a,b}$ is scalar. Indeed, the phase operators $E_{a,0}$ have distinct joint eigenvalue functions on distinct basis states, by nondegeneracy of the trace pairing. Commuting with all of them forces an operator to be diagonal. Commuting with all translations $E_{0,b}$ then forces its diagonal entries to agree.

Direct calculation gives
\begin{equation}\label{eq:weil-weyl-covariance}
\begin{aligned}
    W_{n(t)}E_{a,b}W_{n(t)}^*&=E_{a+tb,b},\\
    W_{h(u)}E_{a,b}W_{h(u)}^*&=E_{ua,u^{-1}b},\\
    W_wE_{a,b}W_w^*&=E_{b,-a}.
\end{aligned}
\end{equation}
For the first identity, the new phase is
$\psi\bigl(a(x+b/2)+t((x+b)^2-x^2)/2\bigr)
=\psi\bigl((a+tb)(x+b/2)\bigr)$.
The second follows by relabelling $x$. For the third, the matrix coefficient of $|y\rangle$ in $W_wE_{a,b}W_w^*|x\rangle$ is
\[
    \frac1p\sum_{z\in\F_p}
    \psi\bigl((a+y-x)z+ab/2+by\bigr).
\]
Character orthogonality makes this zero unless $y=x-a$, in which case it is $\psi\bigl(b(x-a/2)\bigr)$, as required.

The label changes in \cref{eq:weil-weyl-covariance} are exactly multiplication of the column vector $(a,b)^T$ by $n(t)$, $h(u)$, and $w$. These matrices generate $\SL_2(\F_p)$ by \cref{eq:matrix-factorization}. A word for $g$ therefore implements $E_{a,b}\mapsto E_{a',b'}$, where $(a',b')^T=g(a,b)^T$. Two words for the same matrix give the same conjugation action on every $E_{a,b}$, so their operators differ by a scalar. Choosing one word for each $g$ now gives $W_gW_h=c(g,h)W_{gh}$, with $|c(g,h)|=1$. This proves the theorem over every odd prime-power field.
\end{proof}

We call these operations the \emph{quadratic phase} (also called a chirp), \emph{scaling}, and \emph{Fourier transform}. They suffice for every group element. Indeed, the matrix factorization in \cref{eq:matrix-factorization} reads
\begin{equation}\label{eq:bruhat-factorizations}
    \begin{pmatrix}a&b\\c&d\end{pmatrix}
    =\begin{cases}
        n(a/c)\,w\,h(-c)\,n(d/c),&c\neq0,\\
        h(a)n(b/a),&c=0.
    \end{cases}
\end{equation}
This is the two-by-two Bruhat factorization. It expresses every $W_g$, up to a scalar phase, as a product of at most four of the operators above.

\subsection{Why we restrict to the odd subspace}

Define the \emph{parity operator} by $R|x\rangle:=|-x\rangle$. All three operations commute with $R$. Thus the full register has even and odd invariant subspaces, and we cannot apply the irreducible transfer lemma to it. We use the \emph{odd subspace}
\[
    \mathcal H_p^-:=\ker(R+\mathds{1}).
\]
Choose a set $\mathcal R_p\subseteq\F_p^*$ containing exactly one representative of each pair $\{x,-x\}$.

\begin{lemma}[The odd subspace]\label{lem:odd-subspace}
The subspace $\mathcal H_p^-$ is invariant under the projective Weil representation. It has the orthonormal basis
\begin{equation*}
    v_x:=\frac{|x\rangle-|-x\rangle}{\sqrt2}
    \quad(x\in\mathcal R_p),\qquad
    \dim\mathcal H_p^-=\frac{p-1}{2}=N.
\end{equation*}
\end{lemma}
\begin{proof}
The quadratic phase depends on $x^2$, and scaling sends $-x$ to $-u^{-1}x$. Both therefore commute with $R$. In the Fourier sum, substituting $y\mapsto-y$ gives $W_wR=RW_w$. The factorization above now proves invariance under every $W_g$.

An odd vector has zero coefficient at $0$ and opposite coefficients at $x$ and $-x$. The nonzero field elements form $(p-1)/2$ disjoint pairs, so the vectors $v_x$ form the claimed orthonormal basis.
\end{proof}

\subsection{Irreducibility: first diagonal, then scalar}

\begin{lemma}[Irreducibility of the odd Weil representation]\label{lem:odd-weil}
The projective Weil representation on $\mathcal H_p^-$ is irreducible.
\end{lemma}
\begin{proof}
We show that every operator commuting with the representation is scalar. The quadratic phases force it to be diagonal in the odd basis. The scalings then force its diagonal entries to agree.

Let $T\in\End(\mathcal H_p^-)$ commute with every $W_{n(t)}$ and $W_{h(u)}$. Each $v_x$ is a joint eigenvector of the quadratic phases, with eigenvalue function $t\mapsto\psi(tx^2/2)$. Distinct representatives $x,y\in\mathcal R_p$ satisfy $x^2\neq y^2$. By \cref{lem:trace-pairing}, some $t$ therefore satisfies
\[
    \psi\!\left(t(x^2-y^2)/2\right)\neq1.
\]
Hence the joint eigenspaces are exactly the one-dimensional lines $L_{[x]}:=\C v_x$. Commutation with every quadratic phase forces $T$ to preserve each line, so $T$ is diagonal in the odd basis.

Scaling sends $L_{[x]}$ to $L_{[u^{-1}x]}$. For any two nonzero pairs, we can choose $u$ to send the first to the second. Commutation with every scaling therefore forces all diagonal entries of $T$ to agree. The commutant consists of scalars, so \cref{lem:schur-commutant} proves irreducibility.
\end{proof}

The Fourier operator is not needed for this irreducibility argument.

\subsection{Why the choice of lift does not matter}

Our graphs use $G_p=\PSL_2(\F_p)$, while the Weil operators act on lifts in $\SL_2(\F_p)$. The two lifts differ by $-\mathds{1}_2=h(-1)$. The scaling formula gives
\begin{equation*}
    W_{-\mathds{1}_2}=\chi(-1)R,\qquad
    W_{-\mathds{1}_2}\big|_{\mathcal H_p^-}
    =-\chi(-1)\mathds{1}_{\mathcal H_p^-}.
\end{equation*}
Thus changing the lift changes its restriction only by a scalar phase.

For each $g\in G_p$, choose a lift $\widetilde g\in\SL_2(\F_p)$ and a unitary representative of its projective Weil operator. Set
\begin{equation*}
    U_g:=W_{\widetilde g}\big|_{\mathcal H_p^-}.
\end{equation*}
The projective class of $U_g$ is independent of both choices. These operators give an irreducible projective representation of $G_p$, since they include, up to phases, all the restricted Weil operators.

Their conjugation maps give an ordinary representation. Define
\begin{equation*}
    \Gamma_p(g)(X):=U_gXU_g^*,\qquad X\in\End(\mathcal H_p^-).
\end{equation*}
Projective phases cancel, so $\Gamma_p(g)\Gamma_p(h)=\Gamma_p(gh)$. Cyclicity of trace gives
\[
    \langle\Gamma_p(g)(X),\Gamma_p(g)(Y)\rangle_{\HS}
    =\langle X,Y\rangle_{\HS}.
\]
Hence $\Gamma_p$ is a well-defined ordinary unitary representation on operator space. In particular, neither the choice of lift nor the choice of scalar phase affects the channel.

\begin{remark}[Origin of the representation]
The Weil representation is also called the finite oscillator representation. The operators $E_{a,b}$ describe translations and phase modulations. The covariance identities above implement linear changes of their labels that preserve the alternating form $ab'-ba'$. Gurevich, Hadani, and Sochen explain this oscillator picture over odd prime fields~\cite{GurevichHadaniSochen2008}. The proof above supplies the projective representation over every odd prime-power field, with our chosen signs and scaling convention.
\end{remark}

\section{Proof of the Ramanujan bound}\label{sec:quantum-expander}

We now combine the classical graph family with the odd Weil representation. For the generators $S_p$ in \cref{thm:morgenstern-family}, define
\begin{equation*}
    \Phi_p(X):=\frac1D\sum_{s\in S_p}U_sXU_s^*
    =\frac1D\sum_{s\in S_p}\Gamma_p(s)(X).
\end{equation*}

\begin{theorem}[Ramanujan quantum expander]\label{thm:ramanujan-quantum-expander}
For every even $n\geq2$ and $p=q^n$, the map $\Phi_p$ is a degree-$D$ mixed-unitary channel on $\mathcal H_p^-$, of dimension $N=(p-1)/2$. It satisfies
\begin{equation*}
    \left\|\Phi_p\big|_{\End_0(\mathcal H_p^-)}\right\|_{\HS\to\HS}
    \leq\lambda_D.
\end{equation*}
Consequently, every density matrix $\rho$ on $\mathcal H_p^-$ satisfies
\begin{equation*}
    \|\Phi_p(\rho)-\omega_N\|_{\HS}
    \leq\frac{2\sqrt{D-1}}D\|\rho-\omega_N\|_{\HS}.
\end{equation*}
\end{theorem}
\begin{proof}
Each $U_s$ is unitary, so $\Phi_p$ is a mixed-unitary channel. By \cref{thm:morgenstern-family}, the classical walk has norm at most $\lambda_D$ on mean-zero functions. By \cref{lem:odd-weil} and the lift calculation above, $g\mapsto U_g$ is an irreducible projective representation of $G_p$. The transfer lemma gives the claimed bound on traceless operators. The density-matrix statement follows from \cref{lem:spectral-form}.
\end{proof}

\section{Exact circuits and classical preprocessing}\label{sec:quantum-circuit}

We first implement a supplied generator using its factorization into at most four elementary Weil operators. We then assemble the channel. Finally, we show how to compute the field model, generator matrices, and circuit coefficients deterministically. This separates the cost of a channel use from the cost of preparing its circuit description.

\subsection{The register, gate set, and input promise}

The fixed parameters $q=\ell^r$ and $D=q+1$, a model of $\F_q/\F_\ell$, a nonsquare $\nu\in\F_q^*$, and the list $\mathcal T_q$ from \cref{eq:basic-norm-t-generators} are built into the compiler. We can obtain the list by checking all $q^2$ pairs in the fixed field. No $n$-dependent advice is allowed. An $\F_\ell$-operation means an addition, multiplication, inversion, or zero test in this fixed field. The notation $\widetilde O$ suppresses polylogarithmic factors.

We store $\F_p=\F_q[t]/(g_n)$ in the basis obtained from a fixed $\F_\ell$-basis of $\F_q$ and the powers $1,t,\ldots,t^{n-1}$. The register has
\begin{equation}\label{eq:register-size}
    m=rn=\log_\ell p=\log_\ell(2N+1)=\Theta(\log N)
\end{equation}
$\ell$-level qudits. The full register has dimension $p$, while the channel acts on its $N$-dimensional odd subspace $\mathcal H_p^-$. We assume that the input is already supported on this subspace. Encoding and decoding costs are not included.

Our exact gate library contains all computational-basis permutations on at most three $\ell$-level qudits and the one-qudit gates
\[
    F_{\F_\ell}|a\rangle:=\frac1{\sqrt\ell}
        \sum_{b\in\F_\ell}e^{2\pi iab/\ell}|b\rangle,
    \qquad
    Z_\ell|a\rangle:=e^{2\pi ia/\ell}|a\rangle.
\]
We also include controlled versions of these gates and a fixed-size preparation of $D^{-1/2}\sum_{j=1}^D|j\rangle$. The $D$ control labels occupy a subspace of $\lceil\log_\ell D\rceil$ qudits, and the controlled operators act as the identity on unused labels. Since $q,D$, and $\ell$ are fixed, this is a finite gate library independent of $n,p$, and $N$. By an exact \emph{Stinespring circuit} we mean a unitary circuit on the data and an environment, followed by discarding the environment, that realizes the channel exactly. Exactness refers to this declared library, not to exact synthesis over an arbitrary universal gate set.

\begin{theorem}[Uniform exact channel circuit]\label{thm:kraus-complexity}
For every even $n\geq2$, a deterministic compiler outputs an exact Stinespring circuit for $\Phi_p$ over the gate library above. On inputs supported on $\mathcal H_p^-$, the circuit uses
\begin{equation*}
\begin{aligned}
    \text{gate count and depth}&:\quad O(m^2)=O(\log^2N),\\
    \text{ancilla qudits}&:\quad O(m)=O(\log N).
\end{aligned}
\end{equation*}
The preprocessing uses $O(m^5)$ operations over the fixed field $\F_\ell$, or $\widetilde O(m^5)$ bit operations. The hidden constants depend only on $q,\ell,D$, and the fixed local gate library. Encoding into and decoding from $\mathcal H_p^-$ are not included.
\end{theorem}

We prove the circuit bounds next. The preprocessing bound is \cref{lem:uniform-compiler}, proved in \cref{sec:classical-preprocessing}.

\subsection{Implementing the three elementary operators}

For now, suppose we have the field model, a lifted generator, and its factorization coefficients from \cref{eq:bruhat-factorizations}. Each factor is a quadratic phase, a scaling, or a Fourier transform. We implement each with $O(m^2)$ gates and $O(m)$ ancillas.

\paragraph{Reversible field arithmetic.}
The basic arithmetic operations are
\begin{equation*}
\begin{aligned}
    A_a:&\quad |x\rangle\longmapsto|x+a\rangle,\\
    M_a:&\quad |x\rangle\longmapsto|ax\rangle\qquad(a\neq0),\\
    \operatorname{MA}:&\quad |x\rangle|y\rangle|u\rangle
        \longmapsto|x\rangle|y\rangle|u+xy\rangle,\\
    \operatorname{SQ}:&\quad |x\rangle|u\rangle
        \longmapsto|x\rangle|u+x^2\rangle,\\
    T_a:&\quad |x\rangle|c\rangle
        \longmapsto|x\rangle|c+\Tr(ax)\rangle\qquad(c\in\F_\ell).
\end{aligned}
\end{equation*}
Each map is bijective, with inverse given by subtraction or multiplication by $a^{-1}$. Schoolbook polynomial multiplication uses $O(n^2)=O(m^2)$ base-field multiply-adds. Reduction modulo $g_n$ is a linear map $L$ on the unreduced coefficients. We perform it reversibly by retaining those coefficients and applying $(v,u)\mapsto(v,u+Lv)$. This uses $O(m^2)$ controlled additions. Computing into a target and reversing the temporary work gives $O(m^2)$ gates and depth and $O(m)$ ancillas for each map above. Fixed field multiplication and the trace are linear maps in the chosen coordinates and obey the same bounds.

\paragraph{Quadratic phase.}
To apply $W_{n(t)}$, compute $Q_t(x):=\Tr(tx^2/2)$ in an ancilla, apply $Z_\ell$, and reverse the computation:
\begin{equation*}
\begin{aligned}
    |x\rangle|0\rangle
    &\longmapsto |x\rangle|Q_t(x)\rangle\\
    &\longmapsto \psi(tx^2/2)|x\rangle|Q_t(x)\rangle\\
    &\longmapsto \psi(tx^2/2)|x\rangle|0\rangle.
\end{aligned}
\end{equation*}
The arithmetic bounds give $O(m^2)$ gates and depth and $O(m)$ ancillas.

\paragraph{Scaling.}
The operator $W_{h(u)}$ multiplies the field label by the fixed element $u^{-1}$. This is an invertible linear map on $m$ coordinates. Gaussian elimination writes it as $O(m^2)$ elementary coordinate operations, each implemented by a local permutation gate. Its scalar factor $\chi(u)$ is a global phase and may be omitted when we implement the channel.

\paragraph{Fourier transform.}
This is the Fourier transform of the \emph{additive group} of the field, not of $\PSL_2(\F_p)$. In field coordinates it is a linear permutation followed by $m$ fixed-size Fourier transforms. To see this, let $e_1,\ldots,e_m$ be our $\F_\ell$-basis of $\F_p$ and set $B_{ij}:=\Tr(e_ie_j)$. The trace-pairing lemma makes $B$ invertible. For coordinate vectors $x,y$,
\[
    \Tr(xy)=x^TBy=(B^Tx)\cdot y.
\]
Consequently,
\begin{equation*}
    F_{\F_p}=F_{\F_\ell}^{\otimes m}L_{B^T},
    \qquad L_{B^T}|x\rangle:=|B^Tx\rangle.
\end{equation*}
Indeed, the coefficient of $|y\rangle$ on the right is $p^{-1/2}\psi(xy)$, as required. Gaussian elimination implements $L_{B^T}$ with $O(m^2)$ gates. The remaining $m$ gates are copies of $F_{\F_\ell}$. Thus $W_w=F_{\F_p}$ has the claimed cost.

\subsection{From a generator circuit to the channel}

\begin{proof}[Proof of \cref{thm:kraus-complexity}]
By \cref{eq:bruhat-factorizations}, each lifted generator uses at most four elementary Weil operators. The constructions above therefore implement it, up to a scalar phase, with $O(m^2)$ gates and depth and $O(m)$ ancillas. We reuse the workspace between factors.

To apply the channel, choose a generator uniformly, apply its unitary, and forget the choice. This uses one exact $D$-ary random symbol. For a unitary circuit with an environment, use the control register $\mathcal C:=\C^D$ and prepare
\[
    |\Omega_D\rangle:=\frac1{\sqrt D}\sum_{s=1}^D|s\rangle.
\]
Let $\widehat U_s$ be the implemented unitary on the full register $\mathcal K_p$. Its restriction to $\mathcal H_p^-$ equals $U_s$ up to a scalar phase. Apply
\[
    V:=\sum_{s=1}^D|s\rangle\!\langle s|\otimes\widehat U_s
\]
and discard $\mathcal C$. For every input $\rho$ supported on $\mathcal H_p^-$,
\begin{equation*}
    \tr_{\mathcal C}\!\left[
        V(|\Omega_D\rangle\!\langle\Omega_D|\otimes\rho)V^*
    \right]
    =\frac1D\sum_{s=1}^D U_s\rho U_s^*
    =\Phi_p(\rho).
\end{equation*}
The partial trace removes terms with unequal control labels. Each remaining term is unchanged by a scalar phase on its generator. This also explains why we may omit the scalar phases in the elementary operators.

Since $D$ is fixed, controlling all $D$ circuits changes only the constants in the gate and space bounds. Each completed Weil operator preserves $\mathcal H_p^-$, although the primitive arithmetic gates need not preserve it at intermediate times. The preprocessing bound follows from \cref{lem:uniform-compiler} below, and \cref{eq:register-size} converts $m$ to $N$.
\end{proof}

\subsection{Approximation over another gate set}\label{sec:approximate-implementation}

The exact circuit uses the gate set specified above. We now separate two claims about replacing it by another gate set. Local gate synthesis approximates the \emph{encoded channel}, allowing small error outside the odd subspace. To use the mixed-unitary spectral bound below, we also require the approximating Kraus operators to be unitaries on the same odd subspace.

Let $\iota:\mathcal H_p^-\hookrightarrow\mathcal K_p$ be the inclusion and write $\mathcal J(X):=\iota X\iota^*$. We measure output error in the trace norm $\|X\|_1:=\tr\sqrt{X^*X}$, allowing the input to be entangled with any reference register $\mathcal R$.

\begin{lemma}[Approximation of the encoded channel]\label{lem:encoded-approximation}
Fix an inverse-closed finite universal gate set with efficiently computable entries on the constant-size local registers. For $0<\varepsilon<1/2$, the exact circuit in \cref{thm:kraus-complexity} can be replaced by a circuit with
\[
    O\!\left(m^2\log^{1.441}\!\left(\frac{m^2}{\varepsilon}\right)\right)
\]
gates and $O(m)$ ancillas. Let $\widetilde\Phi_{p,\varepsilon}$ denote the resulting channel from the logical input $\mathcal H_p^-$ to the physical output $\mathcal K_p$. For every reference register $\mathcal R$ and density matrix $\rho$ on $\mathcal H_p^-\otimes\mathcal R$,
\[
    \left\|
        (\widetilde\Phi_{p,\varepsilon}\otimes\operatorname{id}_{\mathcal R})(\rho)
        -((\mathcal J\circ\Phi_p)\otimes\operatorname{id}_{\mathcal R})(\rho)
    \right\|_1\leq\varepsilon.
\]
The input is already encoded, as in the exact circuit theorem. The hidden constant also depends on the synthesis gate set.
\end{lemma}
\begin{proof}
Write the full exact Stinespring circuit, including control preparation and workspace, as a product of $L=O(m^2)$ local gates. Kuperberg's Solovay--Kitaev theorem synthesizes a unitary in fixed dimension to operator-norm error $\eta$ with $O(\log^a(1/\eta))$ gates for every fixed $a>\log_\varphi 2=1.44042\ldots$, where $\varphi=(1+\sqrt5)/2$~\cite[Theorem~1.1]{Kuperberg2023BreakingCubic}. We use $a=1.441$ and $\eta=\varepsilon/(2L)$. We synthesize each controlled gate as a whole. Each primitive's scalar phase is then a global phase of the full circuit and may be ignored.

Let $C$ and $\widetilde C$ be the exact and synthesized full circuits. Telescoping their products gives, after choosing a global phase,
\[
    \|\widetilde C-C\|\leq L\eta=\varepsilon/2.
\]
For any input density matrix $\rho_0$, including initialized ancillas and an arbitrary reference register, we therefore have
\[
\begin{aligned}
    \|\widetilde C\rho_0\widetilde C^*-C\rho_0C^*\|_1
    &\leq\|(\widetilde C-C)\rho_0\widetilde C^*\|_1
        +\|C\rho_0(\widetilde C^*-C^*)\|_1\\
    &\leq2\|\widetilde C-C\|\leq\varepsilon.
\end{aligned}
\]
Here the circuits act as the identity on the reference register. Discarding the environment cannot increase this distance. The gate count is $O(L\log^{1.441}(L/\varepsilon))$, and local synthesis uses the same registers.
\end{proof}

The synthesized circuit need not preserve $\mathcal H_p^-$ exactly. Thus the lemma does not assert a spectral bound for a mixed-unitary channel on that subspace. Nor can we use the full physical register instead: the exact full-register channel fixes parity $R$, and hence the nonzero traceless operator $R-p^{-1}\mathds{1}$.

\begin{lemma}[Spectral stability on a fixed logical space]\label{lem:logical-spectral-stability}
Suppose $\widetilde U_s$ are unitaries on $\mathcal H_p^-$ and, after a choice of scalar phases, satisfy $\|\widetilde U_s-U_s\|\leq\delta$ for all $s\in S_p$. Then the degree-$D$ mixed-unitary channel
\[
    \widetilde\Phi(X):=\frac1D\sum_{s\in S_p}
        \widetilde U_sX\widetilde U_s^*
\]
satisfies $\sigma_2(\widetilde\Phi)\leq\lambda_D+2\delta$.
\end{lemma}
\begin{proof}
For every operator $X$ on $\mathcal H_p^-$,
\[
\begin{aligned}
    \|\widetilde U_sX\widetilde U_s^*-U_sXU_s^*\|_{\HS}
    &\leq\|(\widetilde U_s-U_s)X\widetilde U_s^*\|_{\HS}
       +\|U_sX(\widetilde U_s^*-U_s^*)\|_{\HS}\\
    &\leq2\delta\|X\|_{\HS}.
\end{aligned}
\]
Average over $s$ and apply the exact contraction bound to traceless $X$. This gives $\|\widetilde\Phi(X)\|_{\HS}\leq(\lambda_D+2\delta)\|X\|_{\HS}$, as claimed.
\end{proof}

In particular, unitary approximations of error $\varepsilon/2$ give the additive spectral error $\varepsilon$.

\subsection{Deterministic classical preprocessing}\label{sec:classical-preprocessing}

We have shown how to build the circuit from a field model and a list of generator matrices. We finish by computing these data. The only extra condition on the field model is that the residue of $t$ must be a square with a known square root.

\paragraph{Finding a square that generates the field.}
We need a square $\alpha$ that lies in no proper subfield of $\F_{q^n}$ containing $\F_q$. We choose $\alpha=\beta^2$. The danger is that squaring $\beta$ places it in a smaller field. If $F:x\mapsto x^q$ is the Frobenius map, then
\[
    \beta^2\in\F_{q^d}
    \quad\Longrightarrow\quad
    (F^d\beta)^2=\beta^2
    \quad\Longrightarrow\quad
    F^d\beta=\pm\beta.
\]
Thus it suffices to choose $\beta$ outside the kernels of $F^d-\mathds{1}$ and $F^d+\mathds{1}$ for $1\leq d\leq n/2$. These are linear subspaces. We first show that they do not cover the field, then find a point outside them by fixing one coordinate at a time.

\begin{proof}[Proof of \cref{lem:square-irreducible-moduli}]
Shoup's deterministic algorithm constructs a monic irreducible $h\in\F_q[Y]$ of degree $n$ in $\widetilde O(n^4)$ bit operations because $q$ is fixed~\cite[Theorems~3.2 and~4.1]{Shoup1990Irreducible}. Theorem~4.1 covers nonprime base fields. Shoup's notation $x^{\epsilon}$ denotes a fixed polynomial in $\log x$, so his bound has the stated $\widetilde O(n^4)$ form. Work temporarily in $K:=\F_q[Y]/(h)$ and represent the $\F_q$-linear map $F:x\mapsto x^q$ by an $n\times n$ matrix. Put
\[
    W_{d,\varepsilon}:=\ker(F^d-\varepsilon\mathds{1}),
    \qquad 1\leq d\leq n/2,\quad\varepsilon\in\{1,-1\}.
\]
Every element of $W_{d,\varepsilon}$ is a root of $X^{q^d}-\varepsilon X$, so $|W_{d,\varepsilon}|\leq q^d$. Hence
\[
    \sum_{d=1}^{n/2}\sum_{\varepsilon\in\{1,-1\}}
    |W_{d,\varepsilon}|
    \leq n q^{n/2}<q^n=|K|.
\]
The strict inequality follows from $n<q^{n/2}$, valid for even $n\geq2$ and $q\geq3$. In particular, the subspaces do not cover $K$.

To find a point outside their union, fix coordinates on $K$. For a prefix $\pi$, let $C(\pi)$ be the sum, over all the subspaces $W_{d,\varepsilon}$, of the number of their elements extending $\pi$. Gaussian elimination computes each count: an inconsistent system contributes zero, and a solution space of dimension $j$ contributes $q^j$. We have $C(\varnothing)<q^n$ and
\[
    C(\pi)=\sum_{a\in\F_q}C(\pi a).
\]
If $C(\pi)<q^{n-|\pi|}$, at least one choice of $a$ satisfies $C(\pi a)<q^{n-|\pi|-1}$. Choose such an $a$ at each step. After $n$ steps the count is an integer smaller than one, so the resulting $\beta$ lies in none of the subspaces. There are $n$ steps, $q=O(1)$ choices per step, and $n$ subspaces to check. Each count uses an $O(n^3)$ rank computation. The total is $O(n^5)$ field operations.

Set $\alpha:=\beta^2$. The proper subfields of $K$ containing $\F_q$ are the $\F_{q^d}$ with $d\mid n$ and $d<n$, hence $d\leq n/2$. If $\alpha$ lay in one of them, the implication above would give $F^d\beta=\pm\beta$, contrary to our choice. Thus $\alpha$ has degree $n$ over $\F_q$. Also $\alpha\neq0,1$, since $0,1,-1\in W_{1,1}$. Let $g_n$ be the minimal polynomial of $\alpha$ and express $\beta$ in the basis $1,\alpha,\ldots,\alpha^{n-1}$. Its coordinate polynomial $b_n$ satisfies $b_n^2=\overline t$ in $\F_q[t]/(g_n)$.

It remains to find a square root of $\nu$. Since $n$ is even, $K$ contains $\F_{q^2}$. Its nonzero trace-zero elements satisfy $y^q=-y$, so $\ker(F+\mathds{1})$ is nonzero. Choose a nonzero $y$ in this kernel by Gaussian elimination. Then $y^2\in\F_q^*$ and $y^2$ is a nonsquare there. Otherwise a square root in $\F_q$ would force $y\in\F_q$, contradicting $y^q=-y$. Thus $\nu/y^2$ is a square in the fixed field $\F_q$. Enumerating its $q=O(1)$ elements finds $c$ with $c^2=\nu/y^2$. The coordinate polynomial of $cy$ in the $\alpha$-basis is $z_n$, and $z_n^2=\nu$.

Computing $F$, its powers, $g_n$, and the basis conversions uses $O(n^4)$ further operations. Shoup's initial routine also fits within the stated $O(n^5)$ bound for fixed $q$. This proves the lemma.
\end{proof}

\paragraph{Computing the generator matrices and circuit coefficients.}
The square roots now give explicit determinant-one matrices. We then compute the linear and quadratic forms used by the circuits.

\begin{lemma}[Uniform classical compiler]\label{lem:uniform-compiler}
Given $1^n$ for even $n\geq2$, the compiler outputs $g_n$, the $D$ lifted generators in $\SL_2(\F_p)$, their factorization coefficients, the trace-pairing matrix, the quadratic-phase coefficients, and the complete gate lists. It uses $O(m^5)$ operations over the fixed field $\F_\ell$. Its bit complexity is $\widetilde O(m^5)$, and hence at most $O(m^6)$.
\end{lemma}
\begin{proof}
Apply \cref{lem:square-irreducible-moduli}. Since $m=rn$ and $r$ is fixed, this uses $O(m^5)$ operations. Write
\[
    \alpha:=t\bmod g_n,\qquad
    b:=b_n\bmod g_n,\qquad z:=z_n\bmod g_n.
\]
Then $b^2=\alpha$ and $z^2=\nu$. Write $i,j$ for the generators of Morgenstern's quaternion algebra, with $i^2=\nu$, $j^2=t-1$, and $ij=-ji$. Over $\F_p$ we represent them by
\begin{equation}\label{eq:quaternion-splitting}
    \iota(i)=\begin{pmatrix}z&0\\0&-z\end{pmatrix},\qquad
    \iota(j)=\begin{pmatrix}0&1\\\alpha-1&0\end{pmatrix}.
\end{equation}
These matrices square to $\nu\mathds{1}_2$ and $(\alpha-1)\mathds{1}_2$ and anticommute. Since $z\neq0$ and $\alpha\neq1$, the images of $1,i,j,ij$ are linearly independent, so this is an isomorphism onto $M_2(\F_p)$.

For each $(c,d)\in\mathcal T_q$, the image of $\xi_{c,d}$ is
\begin{equation}\label{eq:explicit-generator-matrices}
    M_{c,d}:=\begin{pmatrix}
        1&c+dz\\
        (\alpha-1)(c-dz)&1
    \end{pmatrix}.
\end{equation}
This agrees with Morgenstern's matrix formula after choosing $-z$ as his square root~\cite[Eq.~(14) and Theorem~4.13]{Morgenstern1991Report}. We can check the determinant and inverse pairing directly:
\[
\begin{aligned}
    \det M_{c,d}
        &=1-(\alpha-1)(c^2-\nu d^2)=\alpha=b^2,\\
    M_{c,d}M_{-c,-d}&=\alpha\mathds{1}_2.
\end{aligned}
\]
Consequently,
\begin{equation}\label{eq:uniform-lifts}
    \widetilde s_{c,d}:=b^{-1}M_{c,d}\in\SL_2(\F_p),
    \qquad \widetilde s_{-c,-d}=\widetilde s_{c,d}^{-1},
\end{equation}
and $S_p=\{[M_{c,d}]:(c,d)\in\mathcal T_q\}$. These projective matrices are distinct: their diagonal entries fix the scalar, and their off-diagonal entries recover $c,d$, since $2z(\alpha-1)\neq0$. None is scalar, since that would force $c=d=0$. Computing all $D=O(1)$ lifts uses $O(m^3)$ operations over $\F_\ell$.

The remaining data come from linear algebra. Field inversions and the coefficients in \cref{eq:bruhat-factorizations} cost $O(m^2)$ per generator. We compute the trace as a linear functional in our basis using the Frobenius matrix and its powers, in $O(m^4)$ operations. The trace-pairing matrix $B$ and the $O(1)$ quadratic forms $x\mapsto\Tr(tx^2/2)$ each have $O(m^2)$ coefficients. Once the trace is available, each coefficient costs $O(m^2)$ operations by schoolbook field arithmetic. Their total cost is $O(m^4)$.

Gaussian elimination compiles each invertible linear map into $O(m^2)$ elementary gates in $O(m^3)$ time. For a general linear map $L$, the reversible update $(x,y)\mapsto(x,y+Lx)$ uses $O(m^2)$ controlled additions. The arithmetic constructions above then give the complete gate lists. These costs are dominated by the $O(m^5)$ field construction. Since the base field is fixed, allowing for bit operations and gate indices gives $\widetilde O(m^5)$ bit complexity. The weaker bound $O(m^6)$ removes the suppressed logarithms.
\end{proof}

\begin{corollary}\label{cor:main}
For every fixed odd prime power $q$ and $D=q+1$, put
\[
    p_k:=q^{2k},\qquad N_k:=\frac{q^{2k}-1}{2}\qquad(k\geq1).
\]
The channels $\{\Phi_{p_k}\}_{k\geq1}$ form an infinite uniform family of degree-$D$ Ramanujan quantum expanders in dimensions $N_k\to\infty$. The compiler uses $O(\log^5 N_k)$ operations over a fixed finite field, or $\widetilde O(\log^5 N_k)$ bit operations. Each channel use takes $O(\log^2 N_k)$ gates on inputs already supported on the odd subspace.
\end{corollary}

\DeclareUrlCommand{\Doi}{\urlstyle{sf}}
\renewcommand{\path}[1]{\small\Doi{#1}}
\renewcommand{\url}[1]{\href{#1}{\small\Doi{#1}}}
\bibliographystyle{alphaurl}
\bibliography{refs}

\end{document}